\documentclass{ieeeaccess}
\usepackage{cite}
\usepackage{amsmath,amssymb,amsfonts}
\usepackage{graphicx}
\usepackage{textcomp}
\usepackage{verbatim}
\usepackage[font=small,labelfont=bf,tableposition=top]{caption}
\DeclareCaptionType[fileext=los,placement={th!}]{procedure}

\newcommand{\corcorrecao}{black}

\usepackage{array} 
\usepackage{cite}
\usepackage{multirow} 
\usepackage{algorithm, algorithmic}
\usepackage{graphicx}
\usepackage{subcaption}
\usepackage{verbatim}
\usepackage[table]{xcolor}

\usepackage{multicol}
\usepackage{amsthm} 
\newtheorem{lemma}{Lemma}
\newtheorem{definition}{Definition}
\newtheorem{corollary}{Corollary}

\usepackage{mathrsfs, amsfonts,amsmath,amssymb}
\usepackage{amsthm}

\newcommand{\parametro}{g}
\newcommand{\vbsize}{\mathcal{M}}
\newcommand{\timesample}{Y}
\newcommand{\freqsample}{X}

\def\BibTeX{{\rm B\kern-.05em{\sc i\kern-.025em b}\kern-.08em
    T\kern-.1667em\lower.7ex\hbox{E}\kern-.125emX}}
\begin{document}
\history{Date of publication September 28, 2022.}
\doi{10.1109/ACCESS.2022.3210519}

\title{Is FFT Fast Enough for Beyond 5G Communications? A Throughput-\\Complexity Analysis for OFDM Signals}

\author{\uppercase{Saulo Queiroz}\authorrefmark{1}, 
\uppercase{Joao P. Vilela\authorrefmark{2}, and Edmundo Monteiro}\authorrefmark{3}
\IEEEmembership{Senior Member, IEEE}
}
\address[1]{Academic Department of Informatics,
          Federal University of Techonology (UTFPR), Ponta Grossa, PR, Brazil.
         (e-mail: sauloqueiroz@utfpr.edu.br)}
\address[2]{CISUC, CRACS/INESCTEC, and Dep. of Computer Science, Faculty of Sciences, University of Porto, rua do Campo Alegre s/n, 4169-007 Porto, Portugal (e-mail: jvilela@fc.up.pt)}
\address[3]{CISUC and Department of Informatics Engineering,
University of Coimbra, Portugal 
(e-mail: edmundo@dei.uc.pt)}
\tfootnote{This work is partially supported by the European Regional Development Fund (FEDER), through the Regional Operational Programme of Centre (CENTRO 2020) of the Portugal 2020 framework and FCT under the MIT Portugal Program [Project SNOB-5G with Nr. 045929(CENTRO-01-0247-FEDER-045929)] and the project POWER (grant number POCI-01-0247-FEDER-070365), co-financed by the European Regional Development Fund (FEDER), through Portugal 2020 (PT2020), and by the Competitiveness and Internationalization Operational Programme (COMPETE 2020).}

\markboth
{Is FFT Fast Enough for Beyond 5G Communications? A Throughput-\\Complexity Analysis for OFDM Signals}
{Is FFT Fast Enough for Beyond 5G Communications? A Throughput-\\Complexity Analysis for OFDM Signals}

\corresp{Corresponding author: Saulo Queiroz (e-mail: sauloqueiroz@utfpr.edu.br).}

\begin{abstract}
In this paper, we study the impact of computational complexity on the throughput limits 
of the {\color{black}fast Fourier transform (FFT)} algorithm for {\color{black}orthogonal frequency division multiplexing
(OFDM)} waveforms. 
Based on the spectro-computational {\color{\corcorrecao}complexity} (SC) 
analysis, {\color{\corcorrecao} we verify that the complexity of an $N$-point 
FFT grows faster than the number of bits in the OFDM symbol.} 
Thus, we show that FFT nullifies the OFDM throughput on $N$ unless the
$N$-point discrete Fourier transform (DFT) problem verifies as $\Omega(N)$, which remains a ``fascinating''
open question in theoretical computer science. Also, because FFT 
demands $N$ to be a power of two $2^i$ ($i>0$), the spectrum widening leads 
to an exponential complexity on $i$, i.e. $O(2^ii)$. To overcome these 
limitations, {\color{\corcorrecao} we consider the alternative 
frequency-time transform formulation of vector OFDM (V-OFDM), 
in which an $N$-point FFT is replaced by $N/L$ ($L$$>$$0$) smaller 
{\color{\corcorrecao}$L$-point} FFTs to mitigate the cyclic prefix overhead 
of OFDM. Building on that, we replace FFT by the straightforward DFT algorithm to 
release the V-OFDM parameters from growing as powers of two and to benefit 
from flexible numerology (e.g., $L=3$, $N=156$). Besides, by setting 
$L$ to $\Theta(1)$, the resulting solution can run linearly on $N$
(rather than exponentially on $i$) while sustaining a non null throughput
as $N$ grows.
}
\end{abstract}

\begin{keywords}
Fast Fourier Transform, Computational Complexity, Throughput, Spectro-Computational Analysis, 
Vector OFDM, Parameterized Complexity.
\end{keywords}

\titlepgskip=-15pt

\maketitle
\section{Introduction}  
\IEEEPARstart{T}{he} fast Fourier transform (FFT) algorithm~\cite{CooleyTukey-1965} is among
the top-ten most relevant algorithm of the 20th century~\cite{top10algorithms-2005}.
FFT outperforms the $O(N^2)$ straightforward discrete Fourier transform 
(DFT) algorithm by {\color{\corcorrecao}performing an
$N$-point frequency-time transform} in  $O(N\log_2 N)$ time 
complexity\footnote{{\color{\corcorrecao}
As in the computational complexity theory, by ``time'' or ``time complexity'', we mean 
``number of computational instructions'' unless otherwise stated. The term is interchangeable 
with wall-clock runtime, provided the wall-clock time taken by each instruction on a particular
computational apparatus.}}. 
Particularly for signal communication processing, FFT revolutionized 
the {\color{black} design of an $N$-subcarrier OFDM signal by replacing a bank 
of $N$ synchronized analog} oscillators by a single digital chip that 
requires a single oscillator. 
Ever since, {\color{\corcorrecao}FFT has been employed as the frequency/time transform 
algorithm by several multicarrier and single carrier 
waveforms~\cite{gerzaguet-dspcompcomplexity-2017}}.

\subsection{Motivation and Problem Statement}
{\color{black}
In recent discussions~\cite{madanayake-afft-2020},~\cite{rappaport-100GHz-6g-2019}}, 
scholars have doubted the performance abilities of FFT to modulate signals 
in the future sixth generation~(6G) of wireless networks. 
They point that 6G waveforms are expected to leverage {\color{\corcorrecao}data rate} to 
the order of Terabit per second (Tbit/s) to {\color{\corcorrecao} improve the mobile 
broadband service of 5G}. This envisions signals operating in 
the so-called terahertz (THz) frequency band of the electromagnetic spectrum 
i.e., $0.1$-$10$ $\times 10^{12}$~Hz~\cite{zhao-6g-2019}. 
To alleviate the power consumption implied by the FFT complexity {\color{black}
under wide signals}, Rappaport et al.~\cite{rappaport-100GHz-6g-2019} suggest to 
give up the ``perfect fidelity'' of the DFT computation on behalf of (slightly) 
more error-prone approximation algorithms in portable devices~\cite{madanayake-afft-2020}. 
{\color{black}In other words, the throughput gains envisioned by extremely wideband services
of future wireless networks can lead the computational complexity of FFT to prohibitive 
levels in some practical scenarios. From this, it does result a natural trade-off between 
throughput and complexity that might concern the design of beyond 5G wireless networks.}

{\color{black} In this work, we consider the throughput-complexity trade-off of FFT in the context
of OFDM-based waveforms} and reason about the throughput limit of a DFT algorithm 
considering its computational complexity {\color{black} as the number of points grows}.
In summary, \emph{{\color{black}we place the following questions:
can the FFT complexity impose a  bottleneck that nullifies the OFDM throughput 
as $N$ grows?  Besides, what should be the lower bound asymptotic complexity required 
to sustain a non null throughput of the DFT problem in OFDM?}}

\subsection{Contributions}
 We study the impact of computational
complexity on the throughput limits of different
DFT algorithms (such as FFT) in the context of OFDM-based waveforms.
The spectro-computational (SC) analysis~\cite{queiroz-wcl-19},\cite{queiroz-access-2020},~\cite{queiroz-cost-ixs-19}
is employed to calculate the SC throughput of different DFT 
algorithms. The SC throughput $SC(N)=B(N)/T(N)$ of a signal processing 
algorithm stands for the computational complexity time $T(N)$ spent to modulate
$B(N)$ bits into an $N$-subcarrier symbol. In the SC analysis, a
signal algorithm is asymptotically scalable if its throughput
does not nullify as the spectrum grows, i.e., $\lim_{N\to\infty}SC(N)>0$. 
{\color{black}
Our contributions can be classified into two categories.
First,  we report novel asymptotic limits 
relating complexity and throughput of FFT in the context 
of OFDM signals. Prior works have considered 
the impact of asymptotic complexity on
aspects other than throughput such as DFT silicon area~\cite{thompson-dftarea-79},~\cite{ThompsonBook}
or information loss of computation~\cite{ailon-fft-2015}.
{\color{\corcorrecao} Although complexity and throughput 
have been widely recognized as key performance indicators
for future wireless networks~\cite{zhao-6g-2019}, 
to the best of our knowledge, a formal answer
to our question still lacks in the literature}. In summary,
we demonstrated the following novel asymptotic laws for FFT in
OFDM:
}
\begin{itemize}
  \item The throughput of FFT nullifies on $N$ in
OFDM. Besides, considering that FFT imposes
the number of points to grow as a power of two $N=2^i$ ($i>0$),
{\color{\corcorrecao} the spectrum widening causes FFT to run 
exponentially on $i$, i.e., $O(2^ii)$};
  \item No exact DFT algorithm scales throughput on 
$N$ (given a constellation size $M$) unless the asymptotic complexity lower bound 
of the DFT problem verifies as $\Omega(N)$. Currently, this DFT lower bound remains an 
open ``fascinating'' question in field of computational complexity~\cite{fftlowerbound-2009};
  \item We formalize what we refer to as the sampling-complexity (or the Nyquist-Fourier) trade-off.  
{\color{\corcorrecao}This trade-off} accounts for the fact that the DFT complexity 
increases as the Nyquist interval decreases, {\color{\corcorrecao} causing the
$N$-point DFT computation to become a bottleneck for the sampling task.
Considering OFDM symbols of fixed duration, this trade-off cannot be
solved since it demands a lower bound of $\Omega(1)$ for the DFT problem.}
\end{itemize}

{\color{black}
In our second set of contributions, we consider alternative forms of
frequency-time transform computation under which the resulting complexity $T(N)$
meets the fundamental criterion of the SC analysis $\lim_{N\to\infty}SC(N)>0$, 
i.e., complexity does not nullify throughput as $N$ grows. We disclose how to 
meet such criterion for vector OFDM (V-OFDM)~\cite{vofdm-2001},
a variant of OFDM that replaces an $N$-point FFT by $N/L$ ($L>0$) smaller FFTs 
to mitigate the cyclic prefix overhead of OFDM.
Our contribution results from the fact that other V-OFDM-based 
works e.g.,~\cite{zhang-mimovofdm-access-2020},~\cite{vofdmim-2017},~\cite{cheng-vofdmrayleigh-2011},~\cite{6222369}  
care on aspects other than the throughput-complexity trade-off for the
DFT problem. In this sense, we report the following contributions:
}
\begin{itemize}
  \item {\color{\corcorrecao} We present the SC analysis of the 
frequency-time transform problem in V-OFDM. In this context,
we replace FFT by DFT to relax the power of two constraint on $N$ and
to provide V-OFDM with flexible numerology (e.g. $L=3$, $N=156$). Besides, 
we apply the parameterized complexity technique~\cite{2464827} 
on the DFT algorithm, getting what we refer to as the parameterized DFT (PDFT) 
algorithm. By setting $L=\Theta(1)$, PDFT can run linearly on $N$ rather than
exponentially on $i$ while sustaining a non null throughput as $N$ grows;}
   \item We identify the most efficient setup of V-OFDM to mitigate
sampling-complexity trade-off. By setting $L=2$, PDFT becomes multiplierless
requiring only $O(N)$ complex sums. Although this does not solve the sampling-complexity
trade-off, the most expensive computational instruction of DFT is eliminated and
the additions can be performed in parallel. 
\end{itemize}

The remainder of this work is organized as follows.
In Section~\ref{sec:dftlimits}, we present a joint 
throughput-complexity analysis of the DFT problem
and the FFT algorithm. We also enunciate the 
sampling-complexity (Nyquist-Fourier) trade-off, based
on which we calculate the minimum asymptotic complexity
required for a DFT algorithm to meet the
sampling interval of {\color{black}digital-to-analog/analog-to-digital
(DAC/ADC)} converters in OFDM-based waveforms. In Section~\ref{sec:pdft}, 
we present the PDFT algorithm. In Section~\ref{sec:pdftvalidation},
we present a comparative performance among FFT and the
PDFT algorithm and validate our theoretical results.
In Section~\ref{sec:conclusionpdft}, we summarize the
the work.

\section{Spectro-Computational Asymptotic Analysis}\label{sec:dftlimits}
In this section, we study the joint capacity-complexity asymptotic
limit of the DFT problem by means of the SC analysis (subsection~\ref{subsec:dftlimits}). 
Then, in Subsection~\ref{subsec:fftlimits}, we specialize the analysis to the FFT 
algorithm to respond whether it is sufficiently fast to process signals 
of increasing throughput. Finally, in Subsection~\ref{subsec:sctradeoff}, we
relate the DFT complexity with the Nyquist sampling interval and
introduce what we refer to as the sampling-complexity (Nyquist-Fourier) trade-off.
The notation and symbols used throughout the paper are summarized in Table~\ref{tb:pdftsymbols}.

\begin{table}[]
\centering
\caption{Notation and symbols.\label{tb:pdftsymbols}}
\begin{tabular}{|c|l|}
\hline
Symbol      & Usage \\\hline
$N$         & Number of subcarriers (DFT points)\\ \hline
$\Delta f$  & {\color{\corcorrecao}Subcarrier spacing (Hz)} \\ \hline
$W$         & {\color{\corcorrecao}Signal bandwidth (Hz)} \\ \hline
$M$         & {\color{\corcorrecao}Size of constellation diagram}  \\ \hline
$B(N)$     & Bits per $N$-subcarrier symbol \\ \hline
$T_{DFT}(N)$     & {\color{\corcorrecao}Complexity of a given $N$-point DFT algorithm}\\ \hline
$SC(N)$     & Throughput of an algorithm under an $N$-size input\\ \hline
$T_{NYQ}$     & Inter sample time interval (seconds)\\ \hline
$j$     & Imaginary unity\\ \hline
$\mathbf{\freqsample}$  & Complex frequency domain symbol  \\ \hline
$\mathbf{\timesample}$  & Complex time domain symbol  \\ \hline
$\freqsample_k$     & $k$-th complex frequency domain sample \\ \hline
$\timesample_t$     & $t$-th complex time domain sample \\ \hline
$\mathbf{x}_l$     & $l$-th frequency domain vector block \\ \hline
$\mathbf{y}_q$     & $q$-th time domain vector block \\ \hline
$L$         & Number of vector blocks and DFT size  \\ \hline
$\mathcal{M}$ & Length of vector blocks \\ \hline
$\Omega(f)$ & \begin{tabular}[c]{@{}l@{}}Order of growth asymptotically equal or\\ larger than $f$\end{tabular}  \\ \hline
$O(f)$      & \begin{tabular}[c]{@{}l@{}}Order of growth asymptotically equal or\\ smaller than $f$\end{tabular} \\ \hline
$\Theta(f)$ & Order of growth asymptotically equal to $f$\\ \hline
$[\cdot]^T$  & transpose of the matrix $[\cdot]$\\ \hline
\end{tabular}
\end{table}

\subsection{Capacity-Complexity Limits of the DFT Computation}\label{subsec:dftlimits}
The IDFT at an OFDM transmitter consists in computing the complex 
discrete time samples $\timesample_t$, $t=0,1,\cdots,N-1$ of a symbol given the
complex samples $\freqsample_k$ that modulate the baseband frequencies $k=0,1,\cdots,N-1$. 
According to the Fourier analysis, such relationship is given by 
$\timesample_{t}=\sum_{k=0}^{N-1}\freqsample_{k}e^{j2\pi kt/N}$, $t=0,1,\cdots,N-1$,
in which {\color{\corcorrecao}$j=\sqrt{-1}$}.
At the receiver, a DFT algorithm takes the signal back from time to the
frequency domain by performing 
$\freqsample_{k}=\sum_{t=0}^{N-1}\timesample_{t}e^{-j2\pi kt/N}$, $k=0,1,\cdots,N-1$.
Since in each transform both $k$ and $t$ vary from $0$ to $N-1$, it
is easy to see that the resulting asymptotic complexity $T_{DFT}(N)$ is $O(N^2)$.
The FFT algorithm improves this complexity to $O(N\log_2 N)$ at the constraint 
of $N=2^i$, for some $i>0$. For this reason, the number of FFT points (hence, 
channel width) at least doubles across novel wireless network standards targeting
faster data rates, e.g., IEEE 802.11ax~\cite{80211ax}.
For more details about the theory of DFT and FFT, 
please refer to \cite{kumar-dftsurvey-2019}.

The SC analysis proposed in~\cite{queiroz-wcl-19},~\cite{queiroz-access-2020} 
defines the SC throughput $SC(N)$ bits/time
of an $N$-subcarrier signal processing algorithm
as the ratio between the amount of useful transmission bits $B(N)$ carried by the 
symbol and {\color{\corcorrecao}the overall} computational complexity $T(N)$  {\color{\corcorrecao}required} 
to build the symbol.
For a constellation diagram of size $M=2^p$ (for some $p>0$), each subcarrier 
modulates $\log_2 M$ bits. Thus, in OFDM DFT is performed on a symbol that carries
a total of $B(N)=N\log_2 M$ useful bits.
As usual in the analysis of algorithms, the complexity accounts for the most recurrent 
and expensive computational instruction. Thus, without loss of generality, let now 
$T_{DFT}(N)$ denote the asymptotic number of complex multiplications performed by 
a given DFT algorithm. Let us also denote $T_{MULT}(d)$ as the computational complexity 
to perform a single complex multiplication between two $d$-bit complex numbers. 
For OFDM  $d=\log_2 M$, then the SC throughput of a DFT algorithm in bits/computational 
time is,
\begin{eqnarray}
  SC_{DFT}(N) &=& \frac{N\log_2 M}{T_{MULT}(\log_2 M)T_{DFT}(N)}\label{sc:dftmulti}
\end{eqnarray}
We assume that the channel SNR does not grow arbitrarily on
$N$, meaning that the number of points in the constellation diagram 
is bounded by a constant,  i.e., $M=\Theta(1)$.
Hence, $N$ is the unique variable of our asymptotic 
analysis{\color{\corcorrecao}, i.e., $N\to\infty$\footnote{
Note that this technicality of the asymptotic analysis does not mean the signal
bandwidth is unlimited. Instead, it will enable us to verify whether the
SC throughput nullifies on $N$ through a benefit-cost ratio analysis.}.}
Thus, there exist
constants $d>0$ and $c>0$ such that the SC throughput in (\ref{sc:dftmulti})
rewrites as,
\begin{eqnarray}
  SC_{DFT}(N) &=& \frac{Nd}{T_{DFT}(N)c} \label{sc:dft}
\end{eqnarray}  
Now, proceeding the asymptotic analysis on $N$ and assuming the implied limit exists,
all constants can be neglected and  the following asymptotic SC throughput
results,
\begin{eqnarray}
  SC_{DFT}(N) &=& \lim_{N\to\infty}\frac{N}{T_{DFT}(N)} \label{sc:dftinfty}
\end{eqnarray}

As $N$ grows, both the number of bits of the OFDM symbol as well as the DFT complexity to 
assemble it grows accordingly. The condition for the scalability of a DFT
algorithm as $N$ grows is given in Def.~\ref{def:scadft}.
\begin{definition}[DFT Capacity-Complexity Scalability Condition]\label{def:scadft}
The throughput of a DFT OFDM algorithm of complexity $T_{DFT}(N)$ 
is not scalable unless the inequality~(\ref{eq:conditiondft}) does hold.
\begin{eqnarray}
\lim_{N\to\infty} \frac{N}{T_{DFT}(N)} &>& 0 \label{eq:conditiondft}
\end{eqnarray}
\end{definition}

Based on the fundamental condition of Def.~\ref{def:scadft}, 
the required $T_{DFT}(N)$ time complexity upper-bound is summarized 
in Lemma~\ref{lem:dftrequiredcomplexity}.
\begin{lemma}[Required DFT Asymptotic Complexity for SNR-bounded Channels]~\label{lem:dftrequiredcomplexity}
The throughput of a given $N$-point DFT algorithm 
employed to perform the frequency-time transform 
of a $(d\cdot N)$-bit OFDM waveform ($d>0$) nullifies on $N$ unless 
it runs in $\Theta(N)$ time complexity.
\end{lemma}
\begin{proof}
Let $M$ be the length of the largest constellation diagram 
at which the bit error rate becomes negligible. Assuming the
channel SNR does not grow arbitrarily, $M$ is bounded by
a constant (i.e., $M=\Theta(1)$), so the number of bits $d=\log_2 M$
per subcarrier. Thus, the computational complexity to multiply 
two $d$-bit complex constellation points is bounded 
accordingly, resulting in a constant $c$.
Therefore, the complexity $T_{DFT}(N)$ required
to process $Nd$ bits ensuring the throughput of the 
DFT algorithm does not nullify as $N$ grows (i.e., remains
greater or equal than a non-null constant $k$) is given by:
\begin{eqnarray}
\lim_{N\to\infty} \frac{Nd}{cT_{DFT}(N)} &\geq&  k > 0 \nonumber\\
\lim_{N\to\infty} \frac{Nd}{cT_{DFT}(N)} &\geq& \lim_{N\to\infty} k \nonumber\\
\lim_{N\to\infty}\frac{Nd}{ck} &\geq& \lim_{N\to\infty} T_{DFT}(N) \nonumber\\
T_{DFT}(N) &=& O(N) \label{eq:conditiondftasymptotic}
\end{eqnarray}
Considering the $O(N)$ upper bound of Ineq.~\ref{eq:conditiondftasymptotic}
along with the fact that no $O(N)$ storage complexity DFT algorithm 
can run below $\Omega(N)$ steps~\cite{fftlowerbound-2009} -- assuming that 
at least $N$ computational instructions are needed to read the input --
a non-null throughput DFT algorithm must run in $\Theta(N)$
time complexity.
\end{proof}

If one relaxes the assumption $M=\Theta(1)$ by considering $M$ can grow
faster or as fast as $N$ (i.e., channels of  unbounded SNR), the required 
$T_{DFT}(N)$ complexity upper bound can be calculated from Eq.~\ref{sc:dftmulti} 
by considering either $M=\Omega(N)$ or $M=\Theta(N)$, respectively.
In this case, the overall asymptotic complexity  (so the algorithm throughput)
also depends on the multiplication algorithm, whose complexity depends 
on the number of bits per subcarrier $d=\log_2 M$. Considering, as a matter
of example, $N=\Theta(M)$ and the  $O(d^{\log_2 3})$ complexity of the Karatsuba 
multiplication algorithm~\cite{Karatsuba-62}, the DFT complexity upper bound would 
be nearly $O(N/\log_2^{0.585} N)$. 

\subsection{Spectro-Computational Analysis of the FFT Algorithm}\label{subsec:fftlimits}
The FFT algorithm~\cite{CooleyTukey-1965} outperforms the 
$O(N^2)$ straightforward DFT algorithm by running in $O(N\log_2 N)$ time 
complexity. FFT performs $O(N)$ computational instructions to decrease 
an $N$-point DFT problem into two $N/2$-DFTs per iteration (or recursive calling).
This is possible by noting that the frequency samples $\freqsample_k$ and
$\freqsample_{k+N/2}$ ($k=0,1,\cdots,N/2-1)$ can be computed from the same 
following $N/2$-point DFTs:
\begin{eqnarray}
E_k&=&\sum_{t=0}^{N/2-1}\timesample_{2t}e^{\frac{-j2\pi tk}{N/2}} \label{eq:evendft}\\
O_k&=&e^{-j2\pi k/N}\sum_{t=0}^{N/2-1}\timesample_{2t+1} e^{\frac{-j2\pi tk}{N/2}} \label{eq:odddft}
\end{eqnarray}
 In other words, $E_k$ (Eq.~\ref{eq:evendft}) and $O_k$ (Eq.~\ref{eq:odddft}) 
are the $N/2$-point DFT taken from the even-indexed and odd-indexed time samples 
of the $N$-point input array, respectively. Based on them, the Danielson--Lanczos 
lemma shows that,
\begin{eqnarray}
\freqsample_k&=&E_k + e^{-j2\pi k/N}O_k \label{eq:evendftxk}\\
\freqsample_{k+N/2}&=& E_k - e^{-j2\pi k/N}O_k\label{eq:odddftxkn2}
\end{eqnarray}

This way, $N/2$ iterations are necessary to compute
$\freqsample_k$ and $\freqsample_{k+N/2}$, yielding a total of $O(N)$
computations. Each of these iterations needs to solve both
the $N/2$-point DFTs $E_k$ and $O_k$. Denoting  {\color{\corcorrecao}$T_{DFT}(N)$} as 
the complexity of an $N$-point FFT and applying the  Danielson--Lanczos 
lemma recursively, the overall complexity can be given
by the  recurrence relation {\color{\corcorrecao}$T_{DFT}(N)=O(N) + 2T_{DFT}(N/2)$} which 
results in  {\color{\corcorrecao}$T_{DFT}(N)=O(N\log_2N)$}. Note, however, that FFT 
demands $N=2^i$ ($i>1$), yielding an exponential complexity
of $O(2^i \cdot i)$ on $i$. 
The Corollary~\ref{col:scfft} follows from the $O(N\log_2N)$ complexity of FFT 
in the Lemma~\ref{lem:dftrequiredcomplexity},

\begin{corollary}[Asymptotic Null FFT Throughput] \label{col:scfft}
The spectro-computational throughput of
the FFT algorithm does nullify as $N$ grows.
\end{corollary}
\begin{proof}
From Lemma~\ref{lem:dftrequiredcomplexity}, the FFT throughput follows,
\begin{eqnarray}
 SC_{FFT} &=& \frac{Nd}{cN\log_2N}
\end{eqnarray}
If the SNR can get arbitrarily large such that the constellation
diagram length $M$ grows on $N$ then $d=\Omega(\log_2 N)$. In this case,
the complexity $c$ to multiply two $d$-bit numbers grows at least linearly 
on $d$. Thus, since the fastest multiplying algorithm implies in $c=\Theta(d)$,
the asymptotic throughput of FFT is given by Eq.~\ref{ineq:fftdoesntscales} at best.
Therefore, the FFT throughput nullifies as $N$ grows. 
\begin{eqnarray}
\lim_{N\to\infty}  \frac{N}{N\log_2N} &=& 0 \label{ineq:fftdoesntscales}
\end{eqnarray}
\end{proof}

Fig.~\ref{fig:fftnull} illustrates the asymptotic growth of the
FFT throughput for different subcarrier signal mappers assuming a total
of {\color{\corcorrecao}$T_{DFT}(N)=N\log_2N$} complex multiplications.
{\color{\corcorrecao} Without loss of generality for the asymptotic analysis, 
it is assumed each complex multiplication takes a constant $c_t$ of 1 picosecond, yielding 
a total runtime of $N\log_2N/10^{12}$ seconds. Note that hardware improvements
(e.g., pipelined FFT hardware) translates into lower $c_t$ i.e., faster execution
of a computational instruction. However, the overall \emph{number} of instructions 
remains $N\log_2N$, meaning that better hardware can improve wall-clock runtime
but cannot decrease the complexity of an algorithm. Hence, in fast pipelined FFT hardwares
the complexity penalizes performance indicators other than wall-clock runtime such as
portability, manufacturing cost and power consumption.}
Moreover, for all constellations, widening symbol spectrum by increasing the 
number of subcarriers causes the FFT throughput to decrease rather 
than increasing. This happens because complexity grows faster than 
the number of modulated bits in FFT. To overcome this bottleneck, the  
processing capability of the FFT hardware should scale on $N$.

\begin{figure}[t]
\centering
  \includegraphics[width=3in]{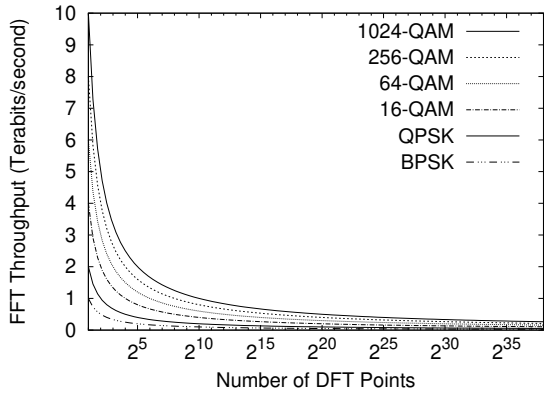}
  \caption{Asymptotic throughput of the FFT algorithm over distinct OFDM signal mappers.
As the number of points increases, complexity grows faster than the number of
modulated bits irrespective of the chosen mapper.
}
  \label{fig:fftnull}
\end{figure}

{\color{\corcorrecao}
We believe that the SC analysis of FFT (as illustrated in Fig.~\ref{fig:fftnull})
formally endorses the issues conjectured by prior works about the infeasibility 
of FFT for some scenarios of future wireless networks. 
In this sense, the FFT throughput nullification implied by complexity
translates the prohibitive power consumption FFT may experience 
under ``massive channel bandwidths''.
}

\subsection{The Sampling-Complexity Nyquist-Fourier Trade-Off}\label{subsec:sctradeoff}
DFT algorithms face two particular issues in the context multicarrier
waveforms such as OFDM. The first comes from a mismatch between the
unit of processing of DFT algorithms and the other algorithms along the
processing block diagram. Although blocks such as ``signal mapping'' and
``cyclic prefix insertion'' process a total of $N$ signal samples, 
they can process them in a sample-by-sample basis. Thus, the processing of 
a particular sample does not depend on the value of other samples in 
those blocks. 

By contrast, DFT algorithms cannot start running
before all $N$ samples are loaded in the input. Hence, the unit of processing
of DFT algorithms is $N$ times higher than their preceding and succeeding  
processing blocks. As $N$ grows, such mismatch turns a DFT algorithm 
to become a bottleneck along the OFDM block diagram. This problem has been 
described by the digital radio design literature as a runtime deadline to
be met by signal processing algorithms~\cite{hessar-usenix-2020},~\cite{liu-realtime5g-2019},~\cite{hellstrom-sdr80211xx-2019},~\cite{drozdenko-hardsoft-2018},~\cite{tan-sora-2011}.
By formalizing the problem as an asymptotic trade-off between sampling and computational overhead,
we can calculate the required asymptotic complexity to meet the sampling interval. 

Second, DFT algorithms are responsible to feed
the DAC in a classic OFDM transmitter. To avoid signal aliasing at the
receiver, the transmitter must sample the time-domain signals produced
by the IDFT algorithm within a specific time interval. This interval
is calculated from the Nyquist sampling theorem which states that
the largest time interval between two equally spaced (time-domain) samples of a 
signal band-limited to $W$~Hz must be $T_{NYQ}=1/(2W)$ seconds. In the case 
of complex IQ modulators where the real and imaginary dimensions of the signal 
are independently and simultaneously sampled by two parallel samplers, 
$T_{NYQ}=1/W$ seconds. 

In IQ systems, \emph{at least} $W$ samples must feed the
DAC every second -- which is known as the Nyquist sampling rate -- otherwise
the signal frequency can suffer from aliasing thereby preventing its 
correct identification at the receiver.
For an inter-subcarrier space of $\Delta f$ Hz, the width of an
 $N$-subcarrier OFDM signal is $W_{OFDM}=N\Delta f$, so a complex time-domain 
OFDM sample must feed the DAC every,
\begin{eqnarray}
T_{NYQ}&=&\frac{1}{N\Delta f}   \quad \text{seconds}\label{eqn:tnyqofdm}
\end{eqnarray}

Based on Eq.~\ref{eqn:tnyqofdm}, we relate the asymptotic complexity 
of DFT algorithms with the Nyquist interval. As result, we introduce the
sampling-complexity (Nyquist-Fourier) trade-Off in Def.~\ref{def:sctradeoff}.
\begin{definition}[The Sampling-Complexity Nyquist-Fourier Trade-Off]\label{def:sctradeoff}
In OFDM radios with $\Delta f$~Hz of inter-subcarrier space, 
the $N$-point DFT computational complexity $T_{DFT}(N)$ increases 
as the Nyquist period $1/(N\Delta f)$ decreases to improve symbol throughput.
\end{definition}
 
The sequence of discrete time samples output by the
IDFT algorithm corresponds to the time-domain version of
the OFDM symbol that lasts $T_{SYM}=1/\Delta f$ seconds. In
the design of a real-time OFDM radio the entire digital signal 
processing must take no more $T_{SYM}$, otherwise the  system either
suffers from sample losses or misses the real-time communication 
capability~\cite{hessar-usenix-2020},~\cite{liu-realtime5g-2019},~\cite{hellstrom-sdr80211xx-2019},
\cite{drozdenko-hardsoft-2018, tan-sora-2011}.
We capture this condition in terms of asymptotic complexity
in Lemma~\ref{lemma:dftnyqupperbound}.

\begin{lemma}[DFT Upper Bound for OFDM Waveforms]\label{lemma:dftnyqupperbound}
The computational complexity upper bound required to solve the 
DFT problem under the Nyquist interval constraint on radios with finite 
processing capabilities is $O(1)$.
\begin{proof}
Considering that a DFT algorithm is the asymptotically most complex
procedure of the basic OFDM waveform, its complexity must satisfy
\begin{eqnarray}
{T_{DFT}(N)} &\leq& T_{SYM} = NT_{NYQ} \label{eqn:dftcond1}
\end{eqnarray}
{\color{\corcorrecao}
Assuming that the throughput improvement is achieved by
enlarging $N$ and that the symbol duration does not grow on $N$ 
(e.g., IEEE 802.11ac~\cite{ieee80211ac-2013}), it follows that}
\begin{eqnarray}
\lim_{N\to\infty} T_{DFT}(N) &\leq& T_{SYM} \label{eq:dftcond2} \\
T_{DFT}(N) &=& O(1) \label{eq:dftcond3}
\end{eqnarray}
\end{proof}
\end{lemma}
Note that one can relax the complexity lower bound predicted
by Lemma~\ref{lemma:dftnyqupperbound} if the radio digital 
baseband processing capabilities can grow arbitrarily on
the number of subcarriers. However, with the end of the
so-called ``Moore's law''~\cite{moore-delayed-2003}, higher
processing capability translates into higher manufacturing 
cost, power consumption and hardware area, bringing doubts to 
the feasibility of portable multicarrier Terahertz radios.

The Corollary~\ref{col:unfeasibledft} follows from Lemma~\ref{lemma:dftnyqupperbound}.
\begin{corollary}[Unfeasible Nyquist-Constrained DFT] \label{col:unfeasibledft}
Given that the minimum possible lower bound complexity of the
DFT problem is $\Omega(N)$~\cite{fftlowerbound-2009} and the
Nyquist interval imposes an upper bound of $O(1)$ (Lemma~\ref{lemma:dftnyqupperbound}), 
\emph{no DFT algorithm can meet the Nyquist interval as $N$ grows}.
\end{corollary}

To face the result of the Corollary~\ref{col:unfeasibledft}, 
one may relax the Nyquist constraint which results in the
compressive sensing systems~\cite{qaisar-compressive-2013}.
However, high accuracy signal frequency prediction  in such systems 
has been proved to be a NP-hard problem~\cite{mousavi-cssurvey-2019} 
which turns out to much more complex systems because only exponential 
time algorithms are known for that class of problems. 

{\color{\corcorrecao} 
Note that the sampling-complexity trade-off does not
restrict to multicarrier waveforms such as OFDM and its variants
but also to single carrier signals that rely on DFT to
mitigate the peak-to-average power ratio of uplink transmissions in
wireless cellular networks~\cite{5gnumerology-2016}. Of course, the
trade-off is more critical to waveforms designed for 
broadband traffic services that target wider spectrum.

}

\section{Pushing the Capacity-Complexity Limits of DFT}\label{sec:pdft}
In this section, we consider methods to overcome 
the throughput bottleneck faced by $N$-point DFT algorithms such as 
FFT (Section~\ref{sec:dftlimits}) {\color{\corcorrecao}and discuss a solution to mitigate
the sampling-complexity trade-off described in Subsection~\ref{subsec:sctradeoff}}.

\subsection{{\color{\corcorrecao}Parameterized} Complexity}\label{subsec:parameterizedcomplexity}
To mitigate the Nyquist-Fourier trade-off in practice,
we apply an algorithm design technique inspired in the parameterized 
complexity~\cite{2464827}.
The parameterized complexity was originally proposed to enable 
the polynomial time solution of multi-parameter NP-complete problems.
The idea consists in bounding one or more parameters of the problem
such that the complexity of the solution becomes a polynomial function 
of the non-bounded parameters. 
For a comprehensive study about the parameterized 
complexity please, refer to~\cite{2464827}. 

{\color{\corcorrecao} We consider an alternative parameterized formulation of the
frequency-time transform problem in order to achieve faster-than FFT computations.}
In a typical OFDM transmitter, the IDFT operation associates $N$ 
input frequency samples $\freqsample_k$ ($k=0,\cdots,N-1$) to $N$ 
respective baseband frequencies $k$~Hz at the time instant $t$ 
by the complex multiplication~$\freqsample_ke^{j2\pi kt/N}$. The direct IDFT algorithm 
repeats these $N$ multiplications to compute $N$ time samples, 
which yields a total of $O(N^2)$ operations. To cut this complexity,
we parameterize the number $\parametro\leq N$ of frequency samples 
associated to a given baseband frequency, as illustrated in Fig.~\ref{fig:pdftschem}.
In the parameterized DFT scheme, all the $N$ frequency 
samples are equally divided across $\parametro$ baseband frequencies $k$, 
leading to $n=N/\parametro$ groups (solid rectangles) of~$\parametro$ 
frequency samples each. 
An $n$-point IDFT across frequency samples of distinct groups (dashed rectangles), 
yields one $\parametro$-sample time domain group per time instant $t=0,1,\cdots,n-1$, 
resulting in a total of $n\parametro=N$ time domain samples.

{\color{\corcorrecao}
We identify that the waveform resulting from the parameterized 
DFT computation we have just described is not new. Indeed,
it exactly matches {\color{\corcorrecao} vector OFDM (V-OFDM),}
a waveform originally proposed to reduce the cyclic prefix 
overhead of OFDM~\cite{vofdm-2001}. Prior works have investigated 
V-OFDM with respect to different aspects. 
Cheng et al.~\cite{cheng-vofdmrayleigh-2011} study the BER
performance in Rayleigh channels and Li et al.~\cite{6222369}
identify setups in which the V-OFDM BER performs similarly or better than OFDM
for different low-complexity receivers. 
More recently,  V-OFDM has been merged with other signal processing techniques
such as index modulation~\cite{vofdmim-2017} and MIMO~\cite{zhang-mimovofdm-access-2020}.  

Our work builds on those prior works to present novel results for
V-OFDM. In particular, we rely on V-OFDM as an alternative to avoid
the throughput nullification faced by FFT in OFDM. Also, we exploit 
the VB structure to relax the power of two constraint of FFT without
giving up a fast asymptotic complexity. By releasing all V-OFDM
parameters from growing as powers of two, more flexible numerologies 
can be enabled (e.g. $n=3$, $N=156$). In Subsections~\ref{subsec:V-OFDM} 
and~\ref{subsec:dftproposal}, we review the V-OFDM signal 
and discuss how to relax the $N=2^i$ constraint of V-OFDM keeping a complexity 
that does not nullify throughput on $N$, respectively.
}

\begin{figure}[t]
        \centering
        \includegraphics[width=3.05in]{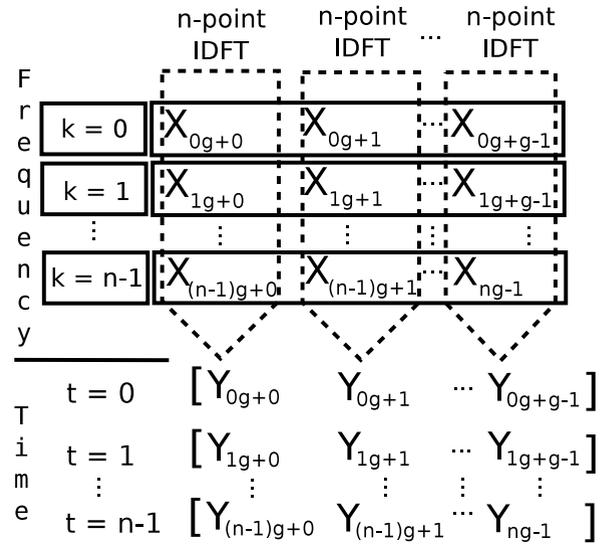}
        \caption{{\color{\corcorrecao}Frequency-time transform scheme of Vector OFDM.
The $N$-size frequency domain input is arranged into $n$ groups (solid rectangles)
of frequency $k$ and length $\parametro=N/n$ each. 
}
\label{fig:pdftschem}}
\end{figure}

\subsection{{\color{\corcorrecao}The Vector OFDM Signal}}\label{subsec:V-OFDM}
The V-OFDM transmitter arranges the $N$-sample complex frequency domain
symbol $\{X_i\}_{i=0}^{N-1}$ into $L$ complex vectors 
blocks (VBs) $\mathbf{x}_{\color{black}l}$ (${\color{black}l=0,1,\cdots,L-1}$) 
having $\vbsize=N/L$ samples each. 
Denoting $[\cdot]^T$ as the transpose of the matrix
$[\cdot]$, the samples of $\{X_i\}_{i=0}^{N-1}$ within the ${\color{black}l}$-th 
VB $\mathbf{x}_l$ is given by
\begin{eqnarray}
\mathbf{x}_{\color{black}l}=[X_{{\color{black}l}\vbsize+\mathsf{m}}]^T\quad \mathsf{m}=0,1,\cdots,\mathcal{M}-1\label{eqn:vbfreqarray}
\end{eqnarray}
The sequence of complex frequency domain samples is
\begin{eqnarray}
\mathbf{\freqsample}=[\freqsample_0,\freqsample_1,\cdots,\freqsample_{N-1}] &=& [\mathbf{x}_{\color{black}0}^T,\mathbf{x}_{\color{black}1}^T,\cdots,\mathbf{x}_{\color{black}{L-1}}^T]\nonumber\\\label{eqn:freqdomainV-OFDMsymbol}
\end{eqnarray}

The ${\color{black}q}$-th time domain 
VB (${\color{black}q=0,1,\cdots,L-1}$) is denoted as
\begin{eqnarray}
\mathbf{y}_{\color{black}q} &=& [\timesample_{{\color{black}q}\cdot\mathcal{M}+\mathsf{m}}]^T\quad \mathsf{m}=0,1,\cdots,\mathcal{M}-1\label{eqn:vbtimearray}
\end{eqnarray}
The V-OFDM literature~\cite{zhang-mimovofdm-access-2020},~\cite{vofdmim-2017},~\cite{cheng-vofdmrayleigh-2011},~\cite{6222369}
performs $\mathcal{M}$ inverse $L$-point FFTs to calculate 
each time domain VB. Since this contrasts to a single 
$N$-point FFT of OFDM, we refer to it as the Parameterized FFT (PFFT).
The resulting samples within the $q$-th time domain VB is therefore
\begin{eqnarray}
\mathbf{y}_{\color{black}q} =
      \begin{bmatrix}
         {\timesample}_{{\color{black}q}\cdot\mathcal{M}+0} \\
         {\timesample}_{{\color{black}q}\cdot\mathcal{M}+1}\\
        \vdots \\ 
         {\timesample}_{{\color{black}q}\cdot\mathcal{M}+(\mathcal{M}-1)} 
     \end{bmatrix} 
           =
      \begin{bmatrix}
         {\freqsample}_{{\color{black}0}\cdot\mathcal{M}+0} \\
         {\freqsample}_{{\color{black}0}\cdot\mathcal{M}+1}\\
        \vdots \\ 
         {\freqsample}_{{\color{black}0}\cdot\mathcal{M}+\mathcal{M}-1} 
     \end{bmatrix}e^{j2\pi{\color{black}q}{\color{black}0}/L} +{\color{black}\cdots}+
     \nonumber\\      
      \begin{bmatrix}
         {\freqsample}_{{\color{black}(L-1)}\mathcal{M}+0} \\
         {\freqsample}_{{\color{black}(L-1)}\mathcal{M}+1}\\
        \vdots \\ 
         {\freqsample}_{{\color{black}(L-1)}\mathcal{M}+\mathcal{M}-1} 
     \end{bmatrix}e^{j2\pi{\color{black}q}{\color{black}(L-1)}/L} 
\label{eqn:timedomainvb}
\end{eqnarray}

The time domain transmitting sequence is 
\begin{eqnarray}
\mathbf{Y}=[\timesample_0,\timesample_1,\cdots,\timesample_{N-1}] = [\mathbf{y}_{\color{black}0}^T,\mathbf{y}_{\color{black}1}^T,\cdots,\mathbf{y}_{\color{black}L-1}^T]\label{eqn:timedomainV-OFDMsymbol}
\end{eqnarray}

Both the normalized inverse DFT and DFT signals are respectively summarized as follows
\begin{eqnarray}
\mathbf{y}_{\color{black}q}&=&\frac{1}{L}\sum_{{\color{black}l=0}}^{{\color{black}L-1}}\mathbf{x}_{\color{black}l}e^{j2\pi{\color{black}q}{\color{black}l}/L}\quad {\color{black}q=0,1,\cdots,L} \\
\mathbf{x}_{\color{black}l}&=&\frac{1}{L}\sum_{{\color{black}q=0}}^{{\color{black}L-1}}\mathbf{y}_{\color{black}q}e^{-j2\pi{\color{black}q}{\color{black}l}/L}\quad {\color{black}l=0,1,\cdots,L}
\end{eqnarray}

After the inverse DFT transform, the signal follows as in the classic OFDM waveform
for transmission. {\color{\corcorrecao}  At the receiver, the signal undergoes the reverse 
steps unless for the detection processing whose complexity can grow exponentially 
on the VB size (e.g., maximum likelihood estimation). While the conditions for
low complexity detection have been discussed by the V-OFDM literature e.g., [15], [30],
in this work we focus on the complexity of the IDFT/DFT problem
in V-OFDM regardless of the chosen detection heuristic. 
In what follows, we adopt the notation $L\mathcal{M}$ (instead of $ng$)
that is usual across the V-OFDM literature.
}

\begin{algorithm}[t]
\begin{algorithmic}[1]
\STATE \COMMENT{$\freqsample_i (i=0,1,\cdots,N-1)$ is the frequency domain input}
\STATE \COMMENT{$\timesample_i (i=0,1,\cdots,N-1)$ is the time domain output}
\STATE \COMMENT{$L$ is the number of points per vector block};
\STATE \COMMENT{$\mathcal{M}$ is the number of vectors such that $N=L\mathcal{M}$};
\FOR {($i=0$; $i<N$; $++i$)}
\STATE $\timesample_i \leftarrow 0$; \COMMENT{initialization of the entire time-domain array};
\ENDFOR
\label{ln:initialization}
\FOR {(${\color{black}q}=0$; ${\color{black}q}<L$; $++{\color{black}q}$)}\label{alg:propdftfor1}
  \FOR {($\mathsf{m}=0$; $\mathsf{m}<\mathcal{M}$; $++\mathsf{m}$)}\label{alg:mfor2}
     \FOR {(${\color{black}l}=0$; ${\color{black}l}<L$; $++{\color{black}l}$)}\label{alg:propdftfor3}
      \STATE 
$\timesample_{{\color{black}q}\cdot \mathcal{M} + \mathsf{m}}=
   \timesample_{{\color{black}q}\cdot \mathcal{M} + \mathsf{m}} + 
     X_{{\color{black}l}\cdot \mathcal{M}  + \mathsf{m}}\cdot e^{j2\pi {\color{black}q}{\color{black}l}/L}$;
     \ENDFOR
  \ENDFOR
\ENDFOR
\end{algorithmic}
  \caption{The parameterized (inverse) DFT (PDFT) algorithm for the Vector OFDM waveform.
By relying on the DFT algorithm and the parameterization technique, 
PDFT relaxes the $N=2^i$ constraint of FFT (thereby enabling
a wider range of numerologies) and can run linearly on $N$ rather
than exponentially on $i$.
\label{alg:pdft}}
\end{algorithm}

\subsection{Parameterized DFT Algorithm for V-OFDM}\label{subsec:dftproposal}
{\color{\corcorrecao}
In order to relax the power of two constraint on the
spectrum parameters of V-OFDM, we replace FFT by the
straightforward DFT algorithm. Since V-OFDM performs $\mathcal{M}$
$L$-point frequency-time transforms, the DFT and FFT
complexities in V-OFDM are $O(\mathcal{M}L^2)$ and $O(\mathcal{M}L\log_2 L)$,
respectively. However, differently from OFDM in which the asymptotic
complexity of DFT cannot be as efficient as FFT's, one can exploit the 
vectorization feature of V-OFDM to refrain the DFT complexity.
 
We achieve that by parameterizing $L$ to $\Theta(1)$, getting what
we refer to as the parameterized DFT (PDFT) algorithm (Algorithm~\ref{alg:pdft}). 
The parameterization provides DFT with non-null throughput on $N$
as demonstrated in Lemma~\ref{lem:pdftV-OFDM}.

\begin{lemma}[Scalable Throughput of the Parameterized DFT Algorithm]~\label{lem:pdftV-OFDM}
By setting the number of points $L$ to $\Theta(1)$,
the PDFT algorithm (Algorithm~\ref{alg:pdft}) achieves
non-null throughput as the number of subcarriers $N$ gets arbitrarily
large (Def.~\ref{def:scadft}). 
\end{lemma}
\begin{proof}
Since $N=ML$, setting $L=\Theta(1)$ leads
the $O(L^2\mathcal{M})$ complexity of PDFT to become $O(\mathcal{M})=O(N)$.
Thus, assuming the channel conditions does not enable 
arbitrarily large constellation diagrams (as in the FFT analysis of Lemma~\ref{lem:dftrequiredcomplexity}),
the total number of bits per V-OFDM symbol is $N\times d=N\times\log_2 M = O(N)$ and the
computational complexity to perform a single complex multiplication
is $\Theta(1)$. Therefore, the throughput (Def.~\ref{def:scadft}) of the 
PDFT algorithm does not nullify on $N$, as demonstrated below:
\begin{eqnarray}
\lim_{N\to\infty} \frac{Nd}{L^2M} =
\lim_{N\to\infty} \frac{Nd}{N} = c > 0 \label{eqn:pdftV-OFDM}
\end{eqnarray}
\end{proof}

If $N=\mathcal{M}L$ is set to grow as a power of two $2^i$, setting
$L$ to $\Theta(1)$ leads both FFT and PDFT to run in $O(\mathcal{M})=O(2^i/L)$ time 
complexity. However, if that constraint is relaxed, PDFT can provide V-OFDM with
flexible numerology while
 running linearly on $N$ rather than exponentially on $i$.
The flexible numerology of PDFT, turns V-OFDM a competitive waveform
for spectrum allocation in fragmented frequency bands. Besides, the
reduced complexity is a step towards the enhancement of current
broadband-driven services such as the enhanced mobile broadband service 
of 5G~\cite{3gpp2020a} and the very high throughput service of 
IEEE 802.11ac~\cite{ieee80211ac-2013}.
}

\begin{figure}[t]
\centering
        \includegraphics[width=3in]{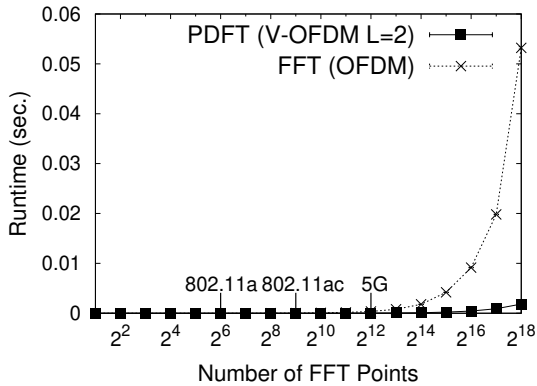}
        \caption{FFT vs. PDFT (proposed): Simulation runtime.  \label{fig:pdftfftruntime}}
\end{figure}

\subsection{Multiplierless Parameterized DFT and the Sampling-Complexity Trade-off}\label{subsec:multiplierless}
{\color{\corcorrecao} 
We identify that the specific case $L=2$ can have 
notable implications for the lower bound complexity of
the frequency-time transform problem in V-OFDM.
As explained in Subsection~\ref{subsec:sctradeoff},
a lower bound complexity of $\Omega(1)$ is required if the frequency-time 
transform computation of a $N$-point signal is constrained by the Nyquist 
sampling theorem. This is typical requisite of real-time implementations
of physical layer standards such as 5G~\cite{liu-realtime5g-2019} and
Wi-Fi~\cite{hessar-usenix-2020},~\cite{hellstrom-sdr80211xx-2019},~\cite{drozdenko-hardsoft-2018},~\cite{tan-sora-2011}.
Next, we explain how the $L=2$ case of  V-OFDM relate to the
sampling-complexity trade-off.

By setting $L$ to $2$, the $N$-subcarrier V-OFDM symbol is
vectorized into only two VBs, leading to $N/2$ 2-point DFTs. Since these
2-point DFTs are completely independent from each other, they can be computed
 in parallel. Each 2-point transform takes $O(1)$ time complexity regardless of 
the value of $N$, therefore the entire solution requires $N/2$ complex additions. 
Indeed, both the indexes $l$ and $q$ that iterate across 
the frequency and time VBs (lines~\ref{alg:propdftfor3} and \ref{alg:propdftfor1} 
in Algorithm~\ref{alg:pdft}, respectively), vary from $0$ to $1$, causing
the complex exponential to simplify to either $1$ or $-1$. The two time domain
VBs are
\begin{eqnarray}
\mathbf{y}_{\color{black}0} &=&
      \begin{bmatrix}
         {\freqsample}_{{\color{black}0}\cdot N/2+0} \\
         {\freqsample}_{{\color{black}0}\cdot N/2+1}\\
        \vdots \\ 
         {\freqsample}_{{\color{black}0}\cdot N/2+N/2-1} 
     \end{bmatrix}e^0
           +
      \begin{bmatrix}
         {\freqsample}_{{\color{black}1}\cdot N/2+0} \\
         {\freqsample}_{{\color{black}1}\cdot N/2+1}\\
        \vdots \\ 
         {\freqsample}_{{\color{black}1}\cdot N/2+N/2-1} 
     \end{bmatrix}e^0\nonumber\\
\label{eqn:multiplierlessvb0}
\end{eqnarray}
and
\begin{eqnarray}
\mathbf{y}_{\color{black}1} &=&
      \begin{bmatrix}
         {\freqsample}_{{\color{black}0}\cdot N/2+0} \\
         {\freqsample}_{{\color{black}0}\cdot N/2+1}\\
        \vdots \\ 
         {\freqsample}_{{\color{black}0}\cdot N/2+N/2-1} 
     \end{bmatrix}e^0 
           +
      \begin{bmatrix}
         {\freqsample}_{{\color{black}1}\cdot N/2+0} \\
         {\freqsample}_{{\color{black}1}\cdot N/2+1}\\
        \vdots \\ 
         {\freqsample}_{{\color{black}1}\cdot N/2+N/2-1} 
     \end{bmatrix}
e^{j\pi}\nonumber\\
\label{eqn:multiplierlessvb1}
\end{eqnarray}
Therefore, from the perspective of an analysis that considers
complex multiplications as the asymptotic dominant instruction of
the DFT problem, the $L=2$ case satisfies the $\Omega(1)$ lower bound
requisite of the sampling-complexity trade-off. By contrast, if all
instructions are considered, the solution does not meet that requisite.
However, the $O(N)$ complex additions are easier to implement in
practice and can remove the DFT bottleneck by being performed in parallel.
Note also that the case $L=1$ of V-OFDM dispenses the DFT computation at the 
transmitter but requires an extra $N$-point IDFT at the receiver. In turn,
the case $L=2$ is multiplierless at both the transmitter and the receiver.
}
\begin{figure}[t]
\centering
      \includegraphics[width=3in]{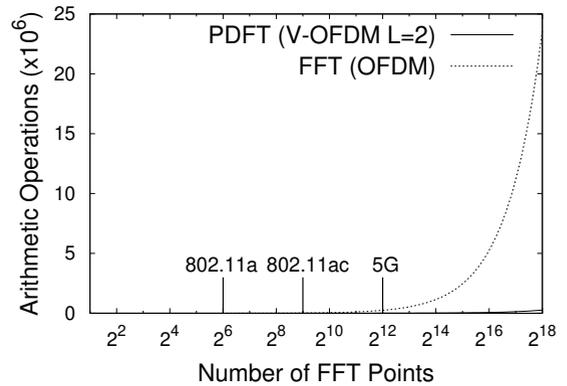}
      \caption{FFT vs. PDFT (proposed):  Complexity.  \label{fig:pdftfftanalyticaltime}}
\end{figure}

\section{Evaluation}\label{sec:pdftvalidation}
In this section, we present simulation results
to compare the FFT and PDFT algorithms
and to validate our theoretical analysis.
{\color{black} Please, note that FFT remains the recommended choice
even for the most recent variations of
V-OFDM~\cite{zhang-mimovofdm-access-2020}\cite{vofdmim-2017}.
Moreover, the FFT performance remains a reference for the
frequency-time transform problem
in current and upcoming wireless network physical layer standards~\cite{rappaport-100GHz-6g-2019},~\cite{3gpp2020a},~\cite{80211ax}. 
}
In Subsection~\ref{subsec:dfttools}, we describe the 
methodology of the simulations. In Subsection~\ref{subsec:poweroftwo},
we discuss the performance of both algorithms under 
a power of two number of points, as required by the FFT 
algorithm. In Subsection~\ref{subsec:dftnonpoweroftwo},
we discuss the performance of the PDFT algorithm under
a non power of two number of points.

\subsection{Tools and Methodology}\label{subsec:dfttools}
We compare our proposed PDFT algorithm for V-OFDM against the FFT
algorithm employed by both OFDM and V-OFDM state of the art. 
We implement the PDFT algorithm in C++ and refer to the
FFT implementation of~\cite{recipes-2007} to assess the FFT algorithm.
It is important to remark that the runtime performance of our chosen
FFT implementation
can be outperformed by highly optimized FFT libraries available
in the literature e.g.,~\cite{1386650}. However, these libraries
impose several {\color{\corcorrecao}preliminary runs} of distinct DFT algorithms to pick
the one that perform best for the considered platform and value
of $N$. Hence, the chosen algorithm may vary across distinct
values of $N$ and the assessed runtime is highly dependent on
several hardware optimizations that vary across the chosen platform. 
By contrast, our focus in this work is on the asymptotic complexity
improvement rather than on hardware optimizations that can be
handled in future work.

\begin{figure}[t]
\centering
      \includegraphics[width=3in]{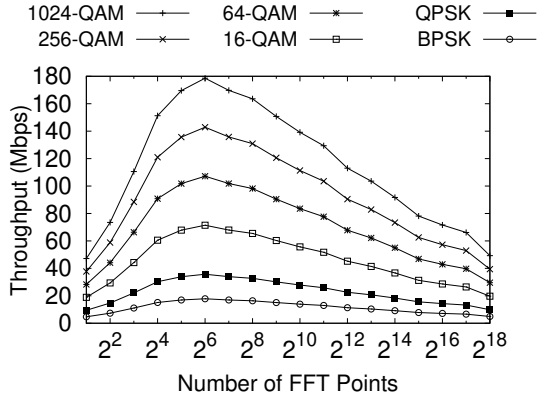}
   \caption{Throughput of FFT algorithm under different signal constellation mappers.  \label{fig:fftthroughput}}
\end{figure}

We vary the number of points which is equivalent
to the number of subcarriers $N$ for both algorithms. 
In this simulation, we vary $N$ as powers of two considering a relatively small number
of subcarriers, as in today's FFT-based waveforms. In the other
simulation, we consider non power of two $N$ and a minimum
of $10^{5}$ subcarriers. In this simulation, we also vary the 
number of VBs of PDFT, as well as the number of points per VB. For each
algorithm, we assess the runtime  {\color{\corcorrecao}$T_{DFT}(N)$} (seconds) and the 
throughput $SC(N)$ (Megabits per second) according to the Def.~\ref{def:scadft}.
We also report the complexity of the compared algorithms. 
{\color{black} Note that the complexity captures the total number of calculations performed
by the algorithms and holds irrespective of the way they are implemented 
(such as pipelined ASIC).}
Unless differently stated, the throughput of each algorithm was
measured considering each subcarrier is BPSK-modulated.

We sampled the wall-clock runtime  {\color{\corcorrecao}$T_{DFT}(N)$} of each algorithm with 
the standard C++ \texttt{timespec} library~\cite{timespec-2018} 
under the profile \texttt{CLOCK\_MONOTONIC} on a 1.8 GHz i7-4500U 
Intel processor with 8 GB of memory. We repeated each experiment 
as many times as needed in order to achieve a mean with relative error
below $5\%$ with a confidence interval of $95\%$. Each sample of 
 {\color{\corcorrecao}$T_{DFT}(N)$} was forwarded to the Akaroa-2 tool~\cite{akaroa2-2010} 
for statistical treatment. Akaroa-2 determined the minimum number of samples 
required to reach the transient-free steady-state mean estimation for  {\color{\corcorrecao}$T_{DFT}(N)$}.
In each execution, we assigned our {\color{black} central processing unit (CPU)} 
process with the largest real-time priority 
and employed the \texttt{isolcpus} Linux kernel directive to allocate one 
physical CPU core exclusively for each process. 
We generate the input points for the algorithms with the standard C++ 64-bit 
version of the Mersenne Twistter (MT) 19937 pseudo-random number 
generator~\cite{matsumoto-mt-1998} set to the seed $1973272912$~\cite{prngs-2002}.
In Tables~\ref{tb:fftpdftruntime} and~\ref{tb:pdftnonpoweroftwo} of
Appendix~\ref{app:tables}, we report the statistics and results discussed in 
subsection~\ref{subsec:poweroftwo} and Subsection~\ref{subsec:dftnonpoweroftwo},
respectively.

\subsection{Power of Two DFTs}\label{subsec:poweroftwo}
In this Subsection, we evaluate the performance of FFT and PDFT 
algorithms under power of two number of points, as required by
the FFT algorithm.
In Fig.~\ref{fig:pdftfftruntime},
we plot the runtime of the FFT algorithm 
(employed by OFDM and V-OFDM)
and the multiplierless PDFT algorithm we propose for V-OFDM set 
to two $N/2$-subcarrier vector blocks.
In Fig.~\ref{fig:pdftfftanalyticaltime}, 
we plot the total number of arithmetic instructions
predicted by the theoretical complexity analysis.
The overall number of  arithmetic instructions performed by
the FFT algorithm and the PDFT algorithm are at least $5N\log_2N$~\cite{1386650} 
and $N$ (Subsection~\ref{subsec:multiplierless}), respectively.
The statistics of the runtime are reported in Table~\ref{tb:fftpdftruntime}. 
We report the throughput considering the BPSK 
modulation in which one bit modulates one subcarrier. Thus, 
one can reproduce Fig.~\ref{fig:fftthroughput} and Fig.~\ref{fig:pdftthroughput} just by
multiplying the BPSK-based throughput with the number of bits
achieved by other modulation, e.g., $6$ in the case of 64-QAM.

\begin{figure}[t]
\centering
   \includegraphics[width=3in]{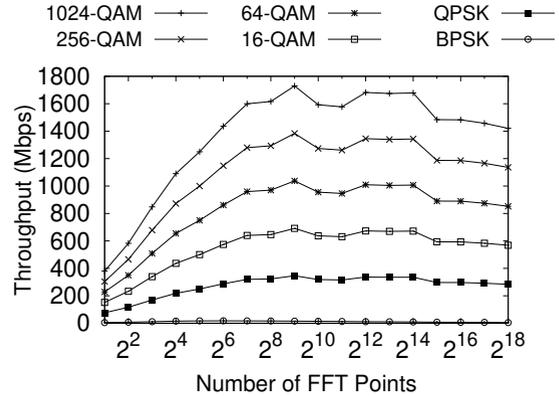}
   \caption{Throughput of PDFT algorithm under different signal constellation mappers.  \label{fig:pdftthroughput}}
\end{figure}

As one can observe in Fig.~\ref{fig:pdftfftruntime}
and Fig.~\ref{fig:pdftfftanalyticaltime}, 
the exponential nature of the FFT complexity
becomes clear after $N=2^{12}=4096$ points. 
Because the FFT algorithm demands $N$ to grow as a power
of two $2^i$ (for some $i>0$), the number of DFT points
must at least double in novel standards that adopt more
subcarriers to improve throughput. Consequently, the
complexity of the FFT algorithm grows accordingly. We
highlight the performance of FFT for the largest number 
of points of different wireless communication standards.
In the case of the IEEE 802.11a~\cite{ieee80211-12}, 
IEEE 802.11ac~\cite{ieee80211ac-2013} and 5G~\cite{3gpp2020a}
{\color{\corcorrecao}physical layer} standards the maximum number of FFT points are 64, 512
and 4096, respectively. {\color{\corcorrecao} Considering the $5N\log_2 N$ 
arithmetic instructions of the Cooley-Tukey algorithm~\cite{1386650}, 
no less than $1920$, $23040$ and $245760$ arithmetic instructions 
must be performed by FFT in those standards, respectively.}
In our simulation, these complexities caused the FFT runtime to grow 
at least one order of magnitude, which corresponded  to 
$3.58$~$\mu s$, $33.97$~$\mu s$ and  $363.8$~$\mu s$, respectively,
as reported in Fig.~\ref{fig:pdftfftruntime}.

The wall-clock runtime of FFT can be improved if FFT is implemented
on dedicate hardware such as {\color{black}application-specific integrated circuits (ASICs)}. However, as
shown in Fig.~\ref{fig:pdftfftanalyticaltime}, the overall number 
of arithmetic instructions remains exponential irrespective of the 
implementation technology. Thus, the FFT complexity represents a 
serious concern for other relevant performance indicators of future 
networks like manufacturing cost, area (device portability) and power 
consumption. 

By contrast, the proposed PDFT algorithm performed 
about two orders of magnitude better than FFT for all scenarios, 
even under the power of two constraint of FFT. Also, the FFT
algorithm nullifies on $N$.
In the simulation, this behavior can be observed by noting
that the FFT throughput reaches the maximum value for $N=2^6$ but
achieves nearly the same value for $N=2^2$ and $N=2^{18}$ (Fig.~\ref{fig:fftthroughput}). 
In turn, the PDFT algorithm keeps nearly the same throughput after 
$N=2^7$ (Fig.~\ref{fig:pdftthroughput}). According to our theoretical analyses, 
this stems from the fact that both the PDFT complexity and the number of
processed bits grows linearly on $N$. Therefore, the PDFT throughput tends to a
non-null constant as $N$ gets arbitrarily large.

\begin{figure}[t]
\centering
        \includegraphics[width=2.5in]{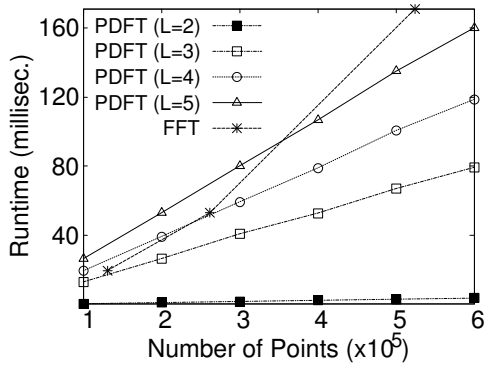}
        \caption{Runtime of FFT and the proposed PDFT 
algorithms for a number of points  $N=1\cdot 10^5, 2\cdot 10^5, \cdots, 6\cdot 10^{5}$. 
For FFT, only the powers of two $2^{17}=131072$, $2^{18}=262144$ and $2^{19}=524288$
are considered.
\label{fig:nonpoweroftworuntime}}
\end{figure}

\subsection{Non Power of Two DFTs}\label{subsec:dftnonpoweroftwo}
In this Subsection, we evaluate the performance of the PDFT 
algorithm under a non power of two number of points $N$.
We vary $N$ through $1\cdot 10^5, 2\cdot 10^{5},\cdots,6\cdot 10^{5}$. 
In Figs.~\ref{fig:nonpoweroftworuntime} and~\ref{fig:nonpoweroftwothroughput}, 
we plot the runtime and
throughput performance of the proposed PDFT algorithm, respectively. 
We vary the number of vector blocks $L=2,3,4,5$ and plot the
performance of the FFT algorithm by setting $N$ to the existing 
powers of two in the interval $[1\cdot 10^5,6\cdot 10^{5}]$, namely 
$2^{17}=131072$, $2^{18}=262144$ and $2^{19}=524288$. 
{\color{\corcorrecao}
PDFT requires the length $N/L$ of each vector block to be an integer.
This requisite is met by all chosen values of $N$ and $L$ except $L=3$.
In this case,  we decrease $N$ by $N \mod 3$ to ensure $N/L$ is an
integer ($x \mod y$ returns the remainder of division of $x$ by $y$). 
Thus, for $L=3$ the values of $N$ $10^5$, $2\cdot 10^5$, $4\cdot 10^5$ 
and $5\cdot 10^5$ are subtracted by $-1$, $-2$, $-1$ and $-2$, respectively.
}
The runtime and throughput of the FFT and PDFT algorithms are
taken from Table~\ref{tb:fftpdftruntime} and Table~\ref{tb:pdftnonpoweroftwo},
respectively. Both tables have the same structure of columns,
as we explained in Subsection~\ref{subsec:poweroftwo}.

\begin{figure}[t]
\centering
      \includegraphics[width=2.5in]{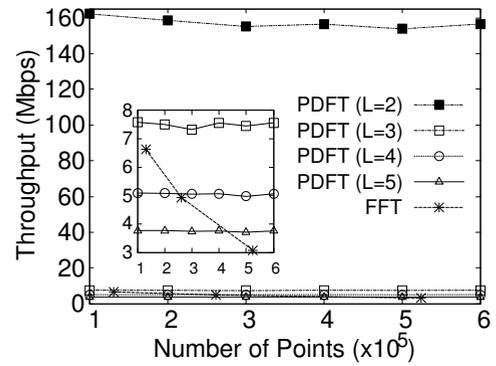}
      \caption{Throughput of FFT and the proposed PDFT 
algorithms for a number of points  $N=1\cdot 10^5, 2\cdot 10^5, \cdots, 6\cdot 10^{5}$. 
For FFT, only the powers of two $2^{17}=131072$, $2^{18}=262144$ and $2^{19}=524288$
are considered.  \label{fig:nonpoweroftwothroughput}}
\end{figure}

As one can see in Fig.~\ref{fig:nonpoweroftworuntime},
the runtime performance of PDFT improves for lower values
of $L$.  The best performance is achieved for $L=2$ in which
PDFT becomes multiplierless and performs $N/2$ $2$-point 
transforms. Although the PDFT performance worsens for 
larger $L$, its complexity remains linear on $N$ for 
all evaluated setups.
This happens because PDFT exploits the parameterization 
technique to perform $\mathcal{M}=N/L$ independent $L$-point DFTs.
By setting $L$ to $\Theta(1)$, each independent DFT
takes $L^2=\Theta(1)$ time complexity, yielding a total of
$(N/L)\cdot \Theta(1)=O(N)$ complexity.

The lowest complexity of PDFT (achieved with $L=2$) translates into the
fastest throughput among all algorithms, which is about
two orders of magnitude above all other algorithms,
as one can see in Fig.~\ref{fig:nonpoweroftwothroughput}
where throughput is plotted considering one bit per point
(i.e., BPSK modulation).
Despite that, PDFT sustains a non-null throughput for
all values of $L$ whereas FFT nullifies as $N$ grows.

The throughput nullification happens because the complexity 
grows asymptotically faster than the number of modulated bits 
as $N$ grows. In the case of PDFT, the throughput remains
constant as $N$ grows even considering the fact that complexity
grows too.
Besides, because PDFT relies on the straightforward DFT algorithm rather 
than FFT, the number of points can grow in an unitary manner rather 
than doubling. Considering the range of the experiment $[1\cdot 10^5, \cdots,6\cdot 10^{5}]$ 
for example, there exist $250001$, $166667$, $125001$ 
and $100001$ possible setup choices of $N$ for PDFT under $L=2$, 
$L=3$, $L=4$ and $L=5$, respectively. By contrast, there are
only three choices of $N$ for the FFT algorithm in the same range,
 they are $2^{17}=131072$, $2^{18}=262144$ and $2^{19}=524288$. 
\emph{This can provide standardization bodies with more setup 
choices for future multicarrier wireless communication standards}.

\section{Conclusion and Future Work}\label{sec:conclusionpdft}
In this work, we demonstrated that the fast Fourier transform
(FFT) algorithm can be too complex for the post-5G generation 
of broadband waveforms. 
The constraint that the number of points $N$ must grow as a power of two 
$2^i$ (for some $i>0$) along with the unprecedented growth in the number of subcarriers, 
cause FFT to run in the exponential complexity $O(2^i\cdot i)$. 
Also, because this complexity grows faster than the number of modulated bits,
the FFT throughput nullifies as $N$ grows.
We generalized this result to show that the throughput of any DFT algorithm
nullifies on $N$ unless the lower bound complexity of the DFT problem
verifies as $\Omega(N)$, which is an open conjecture in computer science. 

{\color{\corcorrecao}
To overcome the scalability limitations of FFT,
we consider the alternative frequency-time transform formulation
of vector OFDM (V-OFDM)~\cite{vofdm-2001},  a waveform that replaces an $N$-point 
FFT by $N/L$ ($L>0$) smaller FFTs to mitigate 
the cyclic prefix overhead of OFDM.
In this sense, we replace FFT by DFT to relax the power of two constraint on $N$ and
to provide V-OFDM with flexible numerology (e.g. $L=3$, $N=156$). Besides, by 
parameterizing $L$ to $\Theta(1)$, we identify that the resulting DFT-based
solution (we refer to as parameterized DFT, PDFT) can run linearly 
on $N$ rather than exponentially on $i$.
}

We also formulate what we refer to as the sampling-complexity 
(Nyquist-Fourier) trade-off, which stems from the fact that
{\color{\corcorrecao} the $N$-point DFT algorithm operates on a batch 
of $N$ samples but its associated sampler operates 
on a sample by sample basis.  As $N$ grows, the Nyquist inter-sample 
time interval demanded by the sampler decreases but the DFT complexity 
to compute all samples increases. We demonstrate that the asymptotic solution 
of the trade-off would require $\Theta(1)$ DFT algorithms. Since
DFT algorithms grows linearly on $N$ at best, i.e., $\Omega(N)$, no DFT 
algorithm can meet the Nyquist deadline as $N$ grows. However,
we identify that the trade-off can be countered in practice if
V-OFDM is set to two $N/2$-subcarrier vector blocks (i.e.,  $L=2$).
In that case, the transform simplifies to $N/2$ complex sums that
can be performed in parallel both at the transmitter and receiver.
Thus, the $N$-point DFT becomes multiplierless and
each sample that feeds the DAC/ADC comes only from two -- rather than $N$ --
other samples. We believe these results turn V-OFDM into a competitive 
candidate waveform for future broadband wireless networks.
}

In future work, the PDFT-based V-OFDM implementation can be 
coupled to an analog Terahertz radio {\color{\corcorrecao} and } the optimal parameterization 
for the PDFT complexity can be identified {\color{\corcorrecao}under} different channel propagation 
conditions.  The joint throughput-complexity asymptotic limit
of detection algorithms can be investigated as well. In this sense,
one may concern about enhancing the analytic model employed 
in this work to capture the natural trade-off between complexity and bit 
error rate in algorithms such as signal detection and error correction codes.
Also, the impact of sampling on the SC analysis of DFT algorithms can be investigated 
under other conditions not considered in this work such as variable symbol duration and
sub-Nyquist samplers~\cite{qaisar-compressive-2013},~\cite{mousavi-cssurvey-2019}.

\appendices
\section{Simulation Results}\label{app:tables}
In  Tables~\ref{tb:fftpdftruntime} and~\ref{tb:pdftnonpoweroftwo}, we 
report the statistics of each simulation. Both tables report the number
of points, the algorithm, the runtime in $\mu s$, the throughput, the
runtime's half-width of the confidence interval and the runtime's variance,
respectively. No experiment demanded more than 70000 repetitions and an
average of about 500 samples were discarded due to the transient stage.

\begin{table}[t]
  \caption{Runtime and throughput of PDFT (V-OFDM, $L=2$) and FFT (V-OFDM) algorithms
under BPSK modulation and power of two number of points. 
$\delta$ is the half-width of the confidence interval with $95\%$ of confidence and relative error below $0.05$.\label{tb:fftpdftruntime}}
\footnotesize
\centering
\begin{tabular}{|l|l|l|l|l|l|}
\hline
\multicolumn{1}{|c|}{N} & \multicolumn{1}{c|}{\textbf{
\begin{tabular}[c]
  {@{}c@{}}Algo-\\ rithm\end{tabular}}} & 
  \multicolumn{1}{c|}{\textbf{\begin{tabular}[c]{@{}c@{}}Runtime \\ $\mu s$\end{tabular}}} & 
  \multicolumn{1}{c|}{\textbf{\begin{tabular}[c]{@{}c@{}}Throughput \\ (Mbps)\end{tabular}}} & 
  \multicolumn{1}{c|}{\textbf{$\pm\delta$}~$\mu s$} & \textbf{Variance} \\ \hline
\multirow{2}{*}{$2^{1}$}  
 & PDFT & \textbf{0.05} & \textbf{38.02} & 0.001 & $<0.001$ \\ \cline{2-6}
 & FFT & \textbf{0.42} & \textbf{4.71} & 0.01 & $<0.001$ \\ \cline{2-6}
\hline
\multirow{2}{*}{$2^{2}$}  
& PDFT & \textbf{0.07} & \textbf{58.31} & 0.001 & $<0.001$ \\ \cline{2-6}
 & FFT & \textbf{0.54} & \textbf{7.35} & 0.03 & $<0.001$ \\ \cline{2-6}
 \hline
\multirow{3}{*}{$2^{3}$} 
 & PDFT & \textbf{0.09} & \textbf{84.84} & 0.001 & $<0.001$ \\ \cline{2-6}
 & FFT & \textbf{0.72} & \textbf{11.06} & 0.03 & $<0.001$ \\ \cline{2-6}
 \hline
\multirow{3}{*}{$2^{4}$}  
& PDFT & \textbf{0.15} & \textbf{109.07} & 0.001 & $<0.001$ \\ \cline{2-6}
 & FFT & \textbf{1.06} & \textbf{15.13} & 0.02 & $<0.001$ \\ \cline{2-6}
 \hline
\multirow{3}{*}{$2^{5}$}  
& PDFT & \textbf{0.26} & \textbf{125.05} & 0.01 & $<0.001$ \\ \cline{2-6}
 & FFT & \textbf{1.89} & \textbf{16.96} & 0.09 & $<0.001$ \\ \cline{2-6}
 \hline
\multirow{3}{*}{$2^{6}$} 
 & PDFT & \textbf{0.45} & \textbf{143.59} & 0.01 & $<0.001$ \\ \cline{2-6}
 & FFT & \textbf{3.58} & \textbf{17.86} & 0.08 & $<0.001$ \\ \cline{2-6}
 \hline
\multirow{3}{*}{$2^{7}$}  
& PDFT & \textbf{0.80} & \textbf{159.96} & 0.01 & $<0.001$ \\ \cline{2-6}
 & FFT & \textbf{7.54} & \textbf{16.97} & 0.37 & 0.02 \\ \cline{2-6}
 \hline
\multirow{3}{*}{$2^{8}$}  
& PDFT & \textbf{1.58} & \textbf{161.66} & 0.08 & $<0.001$ \\ \cline{2-6}
 & FFT & \textbf{15.65} & \textbf{16.36} & 0.51 & 0.05 \\ \cline{2-6}
 \hline
\multirow{3}{*}{$2^{9}$}  
& PDFT & \textbf{2.96} & \textbf{172.94} & 0.01 & $<0.001$ \\ \cline{2-6}
 & FFT & \textbf{33.97} & \textbf{15.07} & 1.26 & 0.29 \\ \cline{2-6}
 \hline
\multirow{3}{*}{$2^{10}$}  
& PDFT & \textbf{6.43} & \textbf{159.24} & 0.30 & 0.02 \\ \cline{2-6}
 & FFT & \textbf{73.58} & \textbf{13.92} & 2.79 & 1.39 \\ \cline{2-6}
 \hline
\multirow{3}{*}{$2^{11}$}  
& PDFT & \textbf{12.99} & \textbf{157.71} & 0.35 & 0.02 \\ \cline{2-6}
 & FFT & \textbf{158.28} & \textbf{12.94} & 0.55 & 0.05 \\ \cline{2-6}
 \hline
\multirow{3}{*}{$2^{12}$}  
& PDFT & \textbf{24.35} & \textbf{168.22} & 0.16 & $<0.001$ \\ \cline{2-6}
 & FFT & \textbf{362.43} & \textbf{11.30} & 2.82 & 1.42 \\ \cline{2-6}
 \hline
\multirow{3}{*}{$2^{13}$}  
& PDFT & \textbf{48.93} & \textbf{167.43} & 0.46 & 0.04 \\ \cline{2-6}
 & FFT & \textbf{790.96} & \textbf{10.36} & 6.01 & 6.45 \\ \cline{2-6}
 \hline
\multirow{3}{*}{$2^{14}$}  
& PDFT & \textbf{97.60} & \textbf{167.87} & 0.18 & 0.01 \\ \cline{2-6}
 & FFT & \textbf{1786.68} & \textbf{9.17} & 3.13 & 1.76 \\ \cline{2-6}
 \hline
\multirow{3}{*}{$2^{15}$}  
& PDFT & \textbf{220.81} & \textbf{148.40} & 0.13 & $<0.001$ \\ \cline{2-6}
 & FFT & \textbf{4193.85} & \textbf{7.81} & 3.55 & 2.25 \\ \cline{2-6}
 \hline
\multirow{3}{*}{$2^{16}$}  
& PDFT & \textbf{442.09} & \textbf{148.24} & 0.38 & 0.03 \\ \cline{2-6}
 & FFT & \textbf{9154.79} & \textbf{7.16} & 60.18 & 647.40 \\ \cline{2-6}
 \hline
\multirow{3}{*}{$2^{17}$}  
& PDFT & \textbf{899.34} & \textbf{145.74} & 6.74 & 8.13\\ \cline{2-6}
 & FFT & \textbf{19805.5} & \textbf{6.62} & 54.58 & 532.5  \\ \cline{2-6}
 \hline

\multirow{3}{*}{$2^{18}$}  
& PDFT & \textbf{1845.65} & \textbf{142.03} &  11.34 & 23.0 \\ \cline{2-6}
 & FFT & \textbf{54415.6} & \textbf{4.82} & 245.92 & 1482 \\ \cline{2-6}
 \hline
\end{tabular}  
\end{table} 

\begin{table}[t]
  \caption{Runtime and throughput of PDFT algorithm 
under BPSK modulation and non power of two number of points. 
$\delta$ is the half-width of the confidence interval with $95\%$ of confidence and relative error below $0.05$.\label{tb:pdftnonpoweroftwo}}
\centering
\footnotesize
\begin{tabular}{|l|l|l|l|l|l|}
\hline
\multicolumn{1}{|c|}{N} & \multicolumn{1}{c|}{\textbf{\begin{tabular}[c]{@{}c@{}}PDFT\\ setup\end{tabular}}} & \multicolumn{1}{c|}{\textbf{\begin{tabular}[c]{@{}c@{}}Runtime \\ $\mu s$\end{tabular}}} & \multicolumn{1}{c|}{\textbf{\begin{tabular}[c]{@{}c@{}}Throughput \\ (Mbps)\end{tabular}}} & \multicolumn{1}{c|}{\textbf{$\pm\delta$}~$\mu s$} & \textbf{Variance} \\ \hline
\multirow{4}{*}{100000}  & L=2 & \textbf{616.27} & \textbf{162.27} & 2.55 & 1.16 \\ \cline{2-6}
 & L=3 & \textbf{13194.70} & \textbf{7.58} & 18.36 & 60.25 \\ \cline{2-6}
 & L=4 & \textbf{19661.60} & \textbf{5.09} & 28.12 & 141.31 \\ \cline{2-6}
 & L=5 & \textbf{26566.00} & \textbf{3.76} & 130.42 & 3040.57 \\ \cline{2-6}\hline
\multirow{4}{*}{200000}  & L=2 & \textbf{1260.86} & \textbf{158.62} & 1.06 & 0.20 \\ \cline{2-6}
 & L=3 & \textbf{26664.20} & \textbf{7.50} & 23.73 & 100.66 \\ \cline{2-6}
 & L=4 & \textbf{39414.40} & \textbf{5.07} & 37.15 & 246.64 \\ \cline{2-6}
 & L=5 & \textbf{53084.40} & \textbf{3.77} & 39.02 & 272.09 \\ \cline{2-6}\hline
\multirow{4}{*}{300000}  & L=2 & \textbf{1933.58} & \textbf{155.15} & 7.33 & 9.60 \\ \cline{2-6}
 & L=3 & \textbf{40969.50} & \textbf{7.32} & 33.04 & 195.16 \\ \cline{2-6}
 & L=4 & \textbf{59452.20} & \textbf{5.05} & 595.27 & 63339.50 \\ \cline{2-6}
 & L=5 & \textbf{80230.30} & \textbf{3.74} & 57.35 & 587.81 \\ \cline{2-6}\hline
\multirow{4}{*}{400000}  & L=2 & \textbf{2556.17} & \textbf{156.48} & 5.00 & 4.46 \\ \cline{2-6}
 & L=3 & \textbf{52958.40} & \textbf{7.55} & 26.59 & 126.36 \\ \cline{2-6}
 & L=4 & \textbf{79045.20} & \textbf{5.06} & 43.66 & 340.75 \\ \cline{2-6}
 & L=5 & \textbf{106685.00} & \textbf{3.75} & 136.26 & 3318.80 \\ \cline{2-6}\hline
\multirow{4}{*}{500000}  & L=2 & \textbf{3250.60} & \textbf{153.82} & 2.05 & 0.75 \\ \cline{2-6}
 & L=3 & \textbf{67125.20} & \textbf{7.45} & 409.26 & 29939.60 \\ \cline{2-6}
 & L=4 & \textbf{100663.00} & \textbf{4.97} & 807.14 & 116450.00 \\ \cline{2-6}
 & L=5 & \textbf{134902.00} & \textbf{3.71} & 969.38 & 167969.00 \\ \cline{2-6}\hline
\multirow{4}{*}{600000}  & L=2 & \textbf{3832.85} & \textbf{156.54} & 3.06 & 1.68 \\ \cline{2-6}
 & L=3 & \textbf{79383.40} & \textbf{7.56} & 29.29 & 153.30 \\ \cline{2-6}
 & L=4 & \textbf{118633.00} & \textbf{5.06} & 57.44 & 589.81 \\ \cline{2-6}
 & L=5 & \textbf{159963.00} & \textbf{3.75} & 294.60 & 15513.50 \\ \cline{2-6}\hline
\end{tabular}  
\end{table}

\bibliographystyle{IEEEtran}
{\color{black}
\bibliography{refs}  %
}
\begin{IEEEbiography}[{\includegraphics[width=1in,height=1.25in,clip,keepaspectratio]{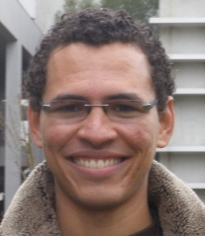}}]{Saulo Queiroz} 
is a professor at the Department of Computer Science of the Federal University of Technology (UTFPR) in Brazil. He completed his Ph.D. with distinction and honor at the University of Coimbra (Portugal). During his academic graduation, he has contributed to open source projects in the field of networking, having participated in initiatives such as as Google Summer of Code. Over the last decade, he has lectured disciplines on computer science such as design and analysis of algorithms, data structures and communication signal processing. His current research interest comprises networking and signal processing for wireless communications.
 \end{IEEEbiography}
\begin{IEEEbiography}[{\includegraphics[width=1in,height=1.25in,clip,keepaspectratio]{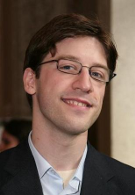}}]{Jo\~ao P. Vilela} 
is an assistant professor at the Department of Computer Science of the University of Porto, Portugal, and a senior researcher at CISUC and INESC TEC. He was a professor at the University of Coimbra after receiving his Ph.D. in Computer Science from the University of Porto in 2011, and a visiting researcher at Georgia Tech and MIT, USA. In recent years, Dr. Vilela has been coordinator and team member of several national, bilateral, and European-funded projects in security and  privacy. His main research interests are in security and privacy of computer and communication systems, with applications such as wireless networks, Internet of Things and mobile devices. Specific research topics include wireless physical-layer security, security of next-generation networks, privacy-preserving data mining, location privacy and automated privacy protection.
\end{IEEEbiography}
\begin{IEEEbiography}[{\includegraphics[width=1in,height=1.25in,clip,keepaspectratio]{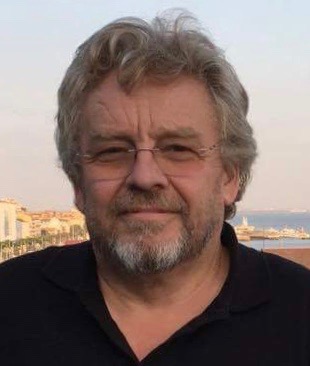}}]{Edmundo Monteiro}
is currently a Full Professor with the University of Coimbra, Portugal.
He has more than 30 years of research experience in the field of computer communications,
wireless networks, quality of service and experience, network and service management, and
computer and network security. He participated in many Portuguese, European, and international
research projects and initiatives. His publication
list includes over 200 publications in journals,
books, and international refereed conferences. He has co-authored nine
international patents. He is a member of the Editorial Board of Wireless
Networks (Springer) journal and is involved in the organization of many
national and international conferences and workshops. He is also a Senior
Member of the IEEE Communications Society and the ACM Special Interest
Group on Communications. He is also a Portuguese Representative in IFIP
TC6 (Communication Systems).
\end{IEEEbiography}

\EOD
\end{document}